\documentclass{llncs}
\usepackage[ruled,vlined]{algorithm2e}

\usepackage{amssymb}
\usepackage{graphicx,color}
\begin{document}
\frontmatter          
\pagestyle{headings}  

\mainmatter              
\title{A Constant Factor Approximation Algorithm for Boxicity of Circular Arc Graphs}
\titlerunning{Algorithms for Boxicity}  
\author{Abhijin Adiga \and Jasine Babu \and L. Sunil Chandran}
\authorrunning{Adiga et.al} 
\institute{Department of Computer Science and Automation,\\Indian Institute of Science, Bangalore 560012, India.\\
\email{\{abhijin, jasine, sunil\}@csa.iisc.ernet.in}}
\maketitle              
\begin{abstract}
Boxicity of a graph $G(V,E)$ is the minimum  integer $k$ such that $G$ can be represented as the intersection graph of $k$-dimensional axis parallel rectangles in $\mathbf{R}^k$. Equivalently, it is the minimum number of interval graphs on the vertex set $V$ such that the intersection of their edge sets is $E$. It is known that boxicity cannot be approximated even for graph classes like bipartite, co-bipartite and split graphs below $O(n^{0.5 - \epsilon})$-factor, for any $\epsilon >0$ in polynomial time unless $NP=ZPP$. Till date, there is no well known graph class of unbounded boxicity for which even an $n^\epsilon$-factor approximation algorithm for computing boxicity is known, for any $\epsilon <1$. In this paper, we study the boxicity problem on Circular Arc graphs - intersection graphs of arcs of a circle. We give a $(2+\frac{1}{k})$-factor polynomial time approximation algorithm for computing the boxicity of any circular arc graph along with a corresponding box representation, where $k \ge 1$ is its boxicity. For Normal Circular Arc(NCA) graphs, with an NCA model given, this can be improved to an additive $2$-factor approximation algorithm. The time complexity of the algorithms to approximately compute the boxicity is $O(mn+n^2)$ in both these cases and in $O(mn+kn^2)= O(n^3)$ time we also get their corresponding box representations, where $n$ is the number of vertices of the graph and $m$ is its number of edges. The additive $2$-factor algorithm directly works for any Proper Circular Arc graph, since computing an NCA model for it can be done in polynomial time.
\end{abstract} 
\section{Introduction} \label{intro}
\textbf{Boxicity: }
Boxicity of a graph $G(V,E)$ is defined as the minimum number of interval graphs on the vertex set $V$ such that the intersection of their edge sets is $E$. If $I_1, I_2, \cdots, I_k$ are interval graphs on the vertex set $V$ such that $G=I_1 \cap I_2 \cap \cdots \cap I_k$, then $\{I_1, I_2, \cdots, I_k\}$ is called a box representation of $G$ and $k$ is called the dimension of the representation. Equivalently, boxicity is the minimum integer $k$ such that $G$ can be represented as the intersection graph of $k$-dimensional axis parallel rectangles in $\mathbf{R}^k$. Boxicity was introduced by Roberts \cite{Rob1} in 1968. If we have a box representation of dimension $k$ for a graph $G$ on $n$ vertices, it can be stored using $O(nk)$ space, whereas an adjacency list representation will need $O(m)$ space which is $O(n^2)$ for dense graphs. The availability of a box representation in low dimension makes some well known NP-hard problems like max-clique, polynomial time solvable\cite{Rosgen}.

Boxicity is combinatorially well studied and many bounds are known in terms of parameters like maximum degree, minimum vertex cover size and tree-width \cite{SUCVC}. Boxicity of any graph is upper bounded by $\lfloor \frac{n}{2}\rfloor$ where $n$ is the number of vertices of the graph. It was shown by Scheinerman \cite{Sch1} in 1984 that the boxicity of outer planar graphs is at most two. In 1986, Thomassen \cite{Thom1} proved that the boxicity of planar graphs is at most 3.  This parameter is also studied in relation with other dimensional parameters of graphs like partial order dimension and threshold dimension \cite{Abh1,Yan1}.

However, computing boxicity is a notoriously hard algorithmic problem. In 1981, Cozzens\cite{Coz1} showed that computing Boxicity is NP-Hard. Later Yannakakis \cite{Yan1}  proved that determining whether boxicity of a graph is at most three is NP-Complete and Kratochvil\cite{Krat1} strengthened this by showing that determining whether boxicity of a graph is at most two itself is NP-Complete. Recently, Adiga et.al \cite{Abh1} proved that no polynomial time algorithm for approximating boxicity of bipartite graphs with approximation factor less than $O(n^{0.5 - \epsilon})$ is possible unless $NP=ZPP$. Same non-approximability holds in the case of split graphs and co-bipartite graphs too. 
Even an $n^\epsilon$-factor approximation algorithm, with $\epsilon < 1$ for boxicity is not known till now, for any well known graph class of unbounded boxicity. In this paper, we present a polynomial time $(2+\frac{1}{k})$-factor approximation algorithm for finding the boxicity of circular arc graphs along with the corresponding box representation, where $k \ge 1$ is the boxicity of the graph. There exist circular arc graphs of arbitrarily high boxicity including the well known Robert's graph (the complement of a perfect matching on $n$ vertices, with $n$ even) which achieves boxicity $\frac{n}{2}$. For normal circular arc graphs, with an NCA model given, we give an additive $2$-factor polynomial time approximation algorithm for the same problem. Note that, proper circular arc graphs form a subclass of NCA graphs and computing an NCA model for them can be done in polynomial time. We also give efficient ways of implementing all these algorithms.\\
\textbf{Circular Arc Graphs:}
Circular Arc (CA) graphs are intersection graphs of arcs on a circle. That is, an arc of the circle is associated with each vertex and two vertices are adjacent if and only if their corresponding arcs overlap. It is sometimes thought of as a generalization of interval graphs which are intersection graphs of intervals on the real line. CA graphs became popular in 1970's with a series of papers from Tucker, wherein he proved matrix characterizations for CA graphs \cite{Tucker1} and structure theorems for some of its important subclasses\cite{Tucker1}. For a detailed description, refer to the survey paper by Lin et.al \cite{Lin09}. Like in the case of interval graphs, linear time recognition algorithms exist for circular arc graphs too \cite{Ross1}. Some of the well known NP-complete problems like tree-width, path-width are known to be polynomial time solvable in the case of CA graphs\cite{Suchan1,Sundaram}. However, unlike interval graphs, problems like minimum vertex coloring \cite{GJ80} and branchwidth \cite{Mazoit06} remain NP-Complete for CA graphs. We believe that boxicity belong to the second category. 

A family $\mathcal{F}$ of subsets of a set $X$ has the Helly property if for every subfamily $\mathcal{F'}$ of  $\mathcal{F}$, with every two sets in $\mathcal{F'}$ pairwise intersecting, we also have $\displaystyle\bigcap_{A \in \mathcal{F'}} A \ne \emptyset$. Similarly, a family $\mathcal{A}$ of arcs satisfy Helly property if every subfamily  $\mathcal{A'}$ $\subseteq$ $\mathcal{A}$ of pairwise intersecting arcs have a common intersection point. The fundamental difficulty while dealing with CA graphs in comparison with interval graphs is the absence of Helly property for a family of circular arcs arising out of their circular adjacencies. 

A Proper Circular Arc (PCA) graph is a graph which has some CA representation in which no arc is properly contained in another. A Unit Circular Arc (UCA) graph is one which has a CA representation in which all arcs are of the same length. A Helly Circular Arc (HCA) graph is one which has a representation satisfying the Helly property. In a CA representation $M$, a pair of arcs are said to be circle cover arcs if they together cover the circumference of the circle. A Normal Circular Arc (NCA) graph is one which has a CA representation in which there are no pairs of circle cover arcs. It is known that UCA $\subsetneq$ PCA $\subsetneq$ NCA and UCA $\nsubseteq$ HCA $\nsubseteq$ NCA.\\
\textbf{Our main results in this paper are:}
\begin{enumerate}
 \item[(a)] Boxicity of any circular arc graph can be approximated within a $(2+\frac{1}{k})$-factor in polynomial time where $k \ge 1$ is the boxicity of the graph.
 \item[(b)] The boxicity of any normal circular arc graph can be approximated within an additive $2$-factor in polynomial time, given a normal circular arc model of the graph.  
 \item[(c)] The time complexity of the algorithms to approximately compute the boxicity is $O(mn+n^2)$ in both the above cases and in $O(mn+kn^2)=O(n^3)$ time we also get their corresponding box representations, where $n$ is the number of vertices of the graph, $m$ its number of edges and $k$ its boxicity.
\end{enumerate}
A structural result we obtained in this paper may be of independent interest. The following way of constructing an auxiliary graph $H^*$ of a given graph $H$ is from \cite{Abu10}. 
\begin{definition} \label{aux}
 Given a graph $H = (V, E)$, consider the graph $H^*$ constructed
as follows: $V(H^*) = E(H)$, and edges $wx$ and $yz$ of $H$ are adjacent in $H^*$ if and only
if $\{w, x, y, z\}$ induces a $2K_2$ in $H$. Notice that $H^*$ is the complement of $[L(H)]^2$, the square of the line graph of $H$.
\end{definition}
 The structural properties of $H^*$ and its complement $[L(H)]^2$ had been extensively investigated for various graph classes in the context of important problems like largest induced matching and  minimum chain cover. The initial results were obtained by Golumbic et.al \cite{Gol2000}. Cameron et.al \cite{Cameron2003} came up with some further results. A consolidation of the related results can be found in  \cite{Cameron2003}. \\\\
The following intermediate structural result in our paper becomes interesting in this context:
\begin{enumerate}
 \item[(d)] In Lemma \ref{lm1}, we observe that if $H$ is a bipartite graph whose complement is a CA graph, then $H^*$ is a comparability graph. 
\end{enumerate}
This is a generalization of similar results for convex bipartite graphs and interval bigraphs already known in literature \cite{Abu10,Yu98}. This observation helps us in reducing the complexity of our polynomial time algorithms.
\section{Preliminaries}
\subsection{Notations}
We denote the vertex set of a given graph by $V(G)$ and edge
set by $E(G)$, with $|V(G)| = n$ and $|E(G)|=m$. We use $e$ to denote $\min(m, nC_2 - m)$. We denote the complement of $G$ by $\overline G$.
We call a graph $G$ the union of graphs $G_1, G_2, \cdots, G_k$ if they are graphs on the same vertex set and $E(G)=E(G_1) \cup E(G_2) \cup \cdots \cup E(G_k)$. Similarly, a graph $G$ is the intersection of graphs $G_1, G_2, \cdots, G_k$ if they are graphs on the same vertex set and $E(G)=E(G_1) \cap E(G_2) \cap \cdots \cap E(G_k)$. We use $box(G)$ to denote boxicity of $G$ and $\chi(G)$ to denote chromatic number of $G$. 

A circular-arc (CA) model $M=(C,\mathcal A)$ consists of a circle $C$, together with a family $\mathcal A$ of arcs of $C$. 
It is assumed that $C$ is always traversed in the clockwise direction, unless stated otherwise. The arc $A_v$ corresponding to a vertex $v$ is denoted by $[s(v), t(v)]$, where $s(v)$ and $t(v)$ are the extreme points of $A_v$ on $C$ with $s(v)$ its start point and $t(v)$ its end point  respectively, in the clockwise direction. Without loss of generality, we assume that no single arc of $\mathcal A$ covers $C$ and no arc is empty or a single point.

An interval model $I$ consists of a family of intervals on real line. An interval $I_v$ corresponding to a vertex $v$ is denoted by a pair $\bigl[l_v(I), r_v(I)\bigr]$, where $l_v(I)$ and $r_v(I)$ are the left and right end points of the interval $I_v$. Without loss of generality, we assume that an interval is always non-empty and is not a single point. We may use $I$ to represent both an interval graph and its interval model, when the meaning is clear from the context. 
\begin{definition}[Bi-Consecutive Adjacency Property]
 Let the vertex set $V(G)$ of a graph $G$ be partitioned into two sets $A$ and $B$ with $|A|=n_1$ and $|B|=n_2$. A numbering scheme where vertices of $A$ are numbered as $1, 2, \cdots, n_1$ and vertices of $B$ are numbered as  $1', 2', \cdots, n_2'$ satisfy Bi-Consecutive Adjacency Property if the following condition holds: \\ For any  $i \in A$ and $j' \in B$,  if $i$ is adjacent to $j'$, then  either \\(a) $j'$ is adjacent to all $k$ such that $1\le k \le i$  or \\(b) $i$ is adjacent to all $k'$ such that $1\le k' \le j'$. 
\end{definition}
\subsection{A Vertex Numbering Scheme for Circular Arc Graphs}\label{Num}
Let $G$ be a CA graph. Assume a CA model $M=(C,\mathcal A)$ of $G$ is given. Let $p$ be any point on the circle $C$. We define a numbering scheme for the vertices of $G$ denoted by $NS(M, p)$ which will be helpful for us in explaining further results.

Let $A$ be the clique corresponding to the arcs passing through $p$ and let $B = V \setminus A$. Let $|A| = n_1$ and $|B| = n_2$. Number the vertices in $A$ as $1, 2, \cdots, n_1$ such that the vertex $v$ with its $t(v)$ farthest (in the clockwise direction) from $p$  gets number $1$ and so on. Similarly, number the vertices in $B$ as  $1', 2', \cdots, n_2'$ such that the vertex $v'$ with its $t(v')$ farthest (in the clockwise direction) from $p$  gets number $1'$ and so on. In both cases, break ties (if any) between vertices arbitrarily, while assigning numbers. See Figure \ref{Fig1} for an illustration of the numbering scheme.
\begin{figure} 
\begin{center}
{\input{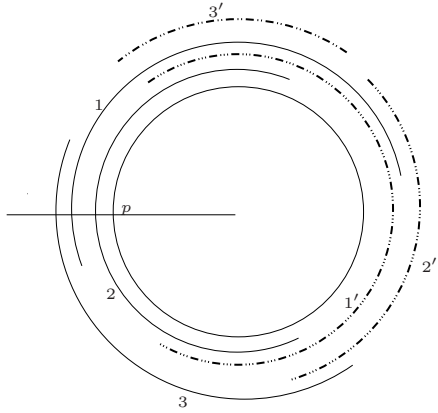}}
\end{center}
\caption{Example for Numbering of vertices of a CA graph}
\label{Fig1}
\end{figure}
Now, observe that in $G$, if a vertex $i \in A$ is adjacent to a vertex $j' \in B$, then at least one of the following is true: 
(a) the point $t(i)$ is contained in the arc $[s(j'), t(j')]$ or (b) the point $t(j')$ is contained in the arc $[s(i), t(i)]$. This implies that if $i \in A$ is adjacent to $j' \in B$, then  either (a) $j'$ is adjacent to all $k$ such that $1\le k \le i$  or (b) $i$ is adjacent to all $k'$ such that $1\le k' \le j'$. Thus we have the following lemma.
\begin{lemma} \label{prop1}
  Given a circular arc graph $G$ and a CA model $M(C, \mathcal{A})$ of $G$, together with a point $p$ on the circle $C$, let $A$ and $B$ be as described above.
\begin{enumerate}
 \item The numbering scheme $NS(M, p)$ of $G$ defined above satisfy the Bi-Consecutive Adjacency Property.  
 \item $NS(M, p)$ can be computed in $O(n^2)$ time.
\end{enumerate}
\end{lemma}
Using Lemma \ref{prop1}, we can prove the following in the case of co-bipartite CA graphs.
 \begin{lemma} \label{CAclique}
 If $G(V, E)$ is a co-bipartite CA graph, then we can find a partition $A\cup B$ of $V$ where $A$ and $B$ induce cliques, having a numbering scheme of the vertices of $A$ and $B$ with $A=\{1, 2, \cdots, n_1\}$ and $B=\{1', 2', \cdots, n_2'\}$ such that it satisfies Bi-Consecutive Adjacency Property. Moreover, the numbering scheme can be found in $O(n^2)$ time. 
\end{lemma}
For a proof of this lemma, refer to Appendix \ref{appendix1}.\\\\
The following lemma is applicable in the case of co-bipartite graphs:
\begin{lemma} \label{lmprop2}
Let $G$ be a co-bipartite graph with a partitioning of vertex set into cliques $A$ and $B$ with $|A| = n_1$ and $|B|=n_2$. Suppose there exist a numbering scheme of vertices of $G$ which satisfies the Bi-Consecutive Adjacency Property. Then $G$ is a CA graph.
\end{lemma}
  The proof is by construction of a CA model $M=(C,\mathcal{A})$ for $G$. Refer to Appendix \ref{appendix1} for the proof. 
\section{Computing the Boxicity of Co-bipartite CA Graphs in Polynomial Time} \label{sspoly}
Using some theorems in the literature, in this section we infer that computing boxicity of co-bipartite CA graphs can be done in polynomial time. A bipartite graph is chordal bipartite if it does not contain any induced cycle of length $\ge 6$.
\begin{theorem}[Feder, Hell and Huang 1999 \cite{Feder99}] \label{th1}
  A graph $G$ is a co-bipartite CA graph if and only if its complement is chordal bipartite and contains no edge-asteroids.
\end{theorem}
A bipartite graph is called a chain graph if it does not contain any induced $2K_2$. The minimum chain cover number of $G$, denoted by $ch(G)$, is the minimum number of chain subgraphs of $G$ such that the union of their edge sets is $E(G)$. \\\\Recall Definition \ref{aux} of $H^*$ from Section \ref{intro}.
\begin{theorem} [Abueida, Busch and Sritharan 2010 \cite{Abu10}] \label{th2}
If $H$ is a bipartite graph with no induced cycles on exactly
6 vertices, then 
\begin{enumerate} 
 \item $ch(H)$ = $\chi(H^*)$.
 \item Every maximal independent set of $H^*$ corresponds to the edge-set of a chain subgraph of $H$. Moreover, the family of maximal independent sets obtained by extending the the color classes of the optimum coloring of $H^*$ corresponds to a minimum chain cover of $H$.
\item In the more restricted case where $H$ is chordal bipartite, $H^*$ is a perfect graph and therefore, $ch(H)$ and a chain cover of $H$ of minimum cardinality can be computed in polynomial time, in view of 1 and 2 above.
\end{enumerate}
\end{theorem}
\begin{theorem}[Yannakakis 1982 \cite{Yan1}] \label{th3} 
 Let $G$ be the complement of a bipartite graph $H$. Then, $box(G)= ch(H)$. Further, if $H_1, H_2, \cdots, H_k$ are chain subgraphs whose union is $H$, their  respective complements $G_1, G_2, \cdots, G_k$ are interval supergraphs of $G$ whose intersection is $G$. 
\end{theorem}
By Theorem \ref{th1}, if $G= \overline H$ is a co-bipartite CA graph, then $H$ is chordal bipartite. Hence by Theorem \ref{th2}, a chain cover of $H$ of minimum cardinality can be computed in polynomial time and $ch(H) = \chi(H^*)$. Combining with Theorem \ref{th3}, we get :
\begin{theorem}\label{th4}
 If $G$ is a co-bipartite CA graph, then $box(G)=\chi(H^*)$ and the family of maximal independent sets obtained by extending the color classes of an optimum coloring of $H^*$ corresponds to the complements of interval supergraphs in an optimal box representation of $G$. Moreover, $box(G)$ and an optimal box representation of $G$ are computable in polynomial time.
\end{theorem}
\section{Reducing the Time Complexity of Computing the Boxicity of Co-bipartite CA Graphs} \label{comp}
 Let $t$ be the number edges of $H$ or equivalently, the number of vertices in $H^*$. By Theorem \ref{th2}, when $H$ is a chordal bipartite graph, $H^*$ is a perfect graph. Using the standard perfect graph coloring methods, an $O(t^3)$ algorithm is given in \cite{Abu10} to compute $\chi(H^*)$. In $O(t^3)$ time, they also compute a chain cover of minimum cardinality. However, $O(t^3)$ can be as bad as $O(n^6)$
in the worst case, where $n$ is the number of vertices of $G$. In \cite{Abu10}, for the restricted case when $H$ is an interval bigraph, they succeeded in reducing the complexity to $O(tn)$, using the zero partitioning property of the adjacency matrix of interval bigraphs. Unfortunately, zero partitioning property being the defining property of interval bigraphs, we cannot use the method used in \cite{Abu10} in our case because of the following result by Hell and Huang \cite{Hell2004}: A graph $H$ is an interval bigraph if and only if its complement is a co-bipartite CA graph admitting a normal CA model. Since there are co-bipartite CA graphs which do not permit a normal CA model, the complements of CA co-bipartite graphs form a strict super class of interval bigraphs. Hence to bring down the complexity of the algorithm from $O(t^3)$, we have to go for a new method. The key ingredient of our method is the following generalization of the results in \cite{Abu10,Yu98}.
\begin{lemma}\label{lm1}
   If the complement of graph $H$ is a co-bipartite CA graph, then $H^*$ is a comparability graph.
\end{lemma}
\begin{proof}
  Let $\overline H = G(V,E)$. Let $A \cup B$ be a partitioning of the vertex set $V$ as described in Lemma \ref{CAclique}, where $A$ and $B$ are cliques. Let $A=\{1, 2, \cdots, n_1\}$ and $B=\{1', 2', \cdots, n_2'\}$ be the associated numbering scheme.

Consider two adjacent vertices of $H^*$ corresponding to the edges $wx'$ and $yz'$ of $H$.  Since they are adjacent, $\{w, x', y, z'\}$ induces a $2K_2$ in $H$. Equivalently, these vertices induce a 4-cycle in $G$ with edges $wy$, $yx'$, $x'z'$ and $z'w$. We claim that $w < y$ if and only if $x' < z'$. To see this, assume that $w < y$. Since $yx' \in E(G)$, by the Bi-Consecutive property of the numbering scheme (Lemma \ref{prop1}), if $z' < x'$, $yz' \in E(G)$ or $wx' \in E(G)$, a contradiction. Hence, $x' < z'$. 

Now, to show that $H^*$ is a comparability graph, we define a relation $\prec$ as $ ab' \prec cd'$ if and only if $a,c \in A$, $b', d' \in B$ with $a < c$ and $b'< d'$ and $\{a, b', c, d'\}$ induces a $2K_2$ in $H$.  In view of the claim proved in the paragraph above, if $ab'$ and $cd'$ are adjacent vertices of $H^*$, they are comparable with respect to the relation $\prec$.

Let $ab' \prec cd'$ and $cd' \prec ef'$.  We have $\{a, b', c, d'\}$ inducing a 4-cycle in G with edges $ac$, $cb'$, $b'd'$ and $d'a$. Similarly, $\{c, d', e, f'\}$ induces a 4-cycle in $G$  with edges $ce$, $ed'$, $d'f'$ and $f'c$. We also have  $a<c<e$ and $b'<d'<f'$, by the definition of the relation $\prec$. By the Bi-Consecutive property of the numbering scheme (Lemma \ref{prop1}), $cf' \in E(G)$ and $cd' \notin E(G)$ implies that $af' \in E(G)$. Similarly, $ed' \in E(G)$ and $cd' \notin E(G)$ implies that $eb' \in E(G)$. Edges $ae$ and $b'f'$ are parts of cliques $A$ and $B$. Hence, we have an induced 4-cycle in $G$ with edges $ae$, $eb'$, $b'f'$ and $f'a$. We can conclude that $ab' \prec ef'$. Thus the relation $\prec$ is transitive and hence, $H^*$ is a comparability graph.
\qed
\end{proof}
\subsection*{Improved Complexities}
Lemma \ref{lm1} serves as the key ingredient in improving the time complexities of our algorithms. By the definition of $H^*$, a proper coloring of the vertices of $H^*$ is same as coloring the edges of $H$ such that no two edges get the same color if their end points induce a $2K_2$ in $H$ or equivalently a 4 cycle in $G$. Since the number of edges in $H^*$ may be of $O(t^2)$, where $t = |E(H)|$, time for computing $\chi(H^*)$ might go up to $O(t^2)=O(n^4)$, if we use the standard algorithm for the vertex coloring of comparability graphs. Let $m_{_{AB}} = n_1 n_2 - t$, the number of edges between $A$ and $B$ in $G$. We show that by utilizing the structure of $G$ along with the underlying comparability relation on the set of non-edges of $G$ defined in the proof of Lemma \ref{lm1}, computing the boxicity of $G$ can be done in $O(en+n^2)$, where  $e$ is $\min(m_{_{AB}}$, $t)$. Each color class can be extended to a maximal independent set and thus get an optimum box representation of $G$ in $O(en+kn^2)$, where $k$ = $box(G)$. The complexities claimed here are obtained by a suitable implementation of the greedy algorithm for the vertex coloring of comparability graphs, fine tuned for this special case and its careful amortized analysis. Due to the structural differences with interval bigraphs as explained before, this turned out to be much different from the method used in \cite{Abu10}. For a detailed description of the algorithm and its analysis, refer to Appendix \ref{complexity}.
\section{Constant Factor Approximation for the Boxicity of CA Graphs}\label{sapprox}
 First we give a lemma which is an adaptation of a similar one given in \cite{Abh1}.
 \begin{lemma} \label{lm4}
  Let $G(V, E)$ be a graph with a partition $(A, B)$ of its vertex set $V$ with $A = \{1, 2, \cdots, n_1\}$ and $B = \{1', 2', \cdots, n'_2\}$. Let $G_1(V, E_1)$ be its supergraph such that $E_1 = E \cup \{(a',b'): a',b' \in B\}$. Then, $box(G_1) \le 2 \cdot box(G)$.
 \end{lemma}
For a proof of this lemma, see Appendix \ref{appendix1}.
\begin{definition}
 Let $G(V, E)$ be an interval graph and $I$ be an interval representation of $G$. Let $l=\displaystyle\min_{u \in V}$ ${l_u(I)}$ and $r=\displaystyle\max_{u \in V}$ ${r_u(I)}$. Consider a graph $G'(V', E')$ such that $V' \supseteq V$ and $E' = E \cup \{(a, b)$: $a \in V' \setminus V$ and $b \in V'\}$. An interval representation $I'$ of $G'$ obtained by assigning interval $[l, r]$, $\forall u \in V' \setminus V$ and intervals $[l_u(I), r_u(I)]$, $\forall u \in V$ is called an extension of $I$ on $V'$.
\end{definition}
\begin{algorithm}
     \LinesNumbered 
     \SetAlgoNoLine
     \DontPrintSemicolon
    \caption{Find a near optimal box representation of given CA graph}
     \label{algApprox}    
     \KwIn{A circular arc graph $G(V, E)$}
     \KwOut{A box representation of $G$ of dimension at most $2k+1$ where $k=box(G)$}    
     \lIf {$G$ is an interval graph}{
        Output an interval representation $I_G$ of $G$, Exit}\;
        Compute a CA model $M(C, \mathcal A)$ of $G$\;
        Choose any point $p$ on the circle $C$\;
       Let $A$ be the clique corresponding to $p$; \hspace{0.25cm}$B = V \setminus A$\;
       Construct $G'(V, E')$ with $E' = E \cup \{(u',v'): u',v' \in B\}$\\\label{g'}\tcc*[f]{$G'$ is a co-bipartite CA graph by Lemma \ref{lmCA}} \;
       Find an optimum box representation $\mathcal{B'}=\{I_1', I_2', \cdots, I_{b}'\}$ of $G'$\\\tcc*[f]{Using the method described in Section \ref{sspoly}}\;
       Construct an interval representation $I$ for the subgraph induced on $B$\\\tcc*[f]{Induced subgraph on B is clearly an interval graph}\;
       Construct $I'$, the extension of $I$ on $V$\label{i'}\;
       Output  $\mathcal{B}=\{I_1', I_2', \cdots, I_{b}', I'\}$ as the box representation of $G$\label{b'}\; 
\end{algorithm}
\subsection*{Approximation Algorithm}
A method for computing a box representation of a given CA graph $G$ within a $(2+\frac{1}{k})$-factor where $k\ge 1$ is the boxicity of $G$ is given in Algorithm \ref{algApprox}. We use the $O(en)$ algorithm for computing boxicity of co-bipartite CA  graphs given in Section \ref{sspoly} as a subroutine here. Let $n=|V(G)|$ and $m=|E(G)|$. We can show that a near optimal box representation of $G$ can be obtained in $O(mn+kn^2)$. For more details, refer to Appendix \ref{appendix1}. If we just want to compute the approximate boxicity of $G$, it is enough to output $box(G')+1$, as proved below. This can be done in $O(mn+n^2)$.\\
\textbf{Proof of correctness: }
 Let us analyze the non-trivial case when $G$ is not an interval graph. Otherwise, the correctness is obvious.
\begin{lemma} \label{lmCA}
 $G'$ constructed in Line \ref{g'} of Algorithm \ref{algApprox} is a co-bipartite CA graph.
\end{lemma}
\begin{proof}
 It can be easily seen that $G'$ is a co-bipartite graph on the same vertex set as that of $G$ with cliques $A$ and $B$ and $V=A \cup B$. 
 Consider a numbering scheme $NS(M, p)$ of $G$ as described in Section \ref{Num} such that $A= \{1, 2, \cdots, n_1\}$ and  $B= \{1', 2', \cdots, n_2'\}$, based on the CA model $M(C, \mathcal A)$ and the point $p$ as chosen in Algorithm \ref{algApprox}. Notice that by construction of $G'$, for any pair of vertices $i \in A$ and $j' \in B$, $(u, v') \in E$ if and only if $(u, v') \in E'$. Recall that the numbering scheme $NS(M, p)$ satisfies Bi-Consecutive Adjacency Property for $G$ by Lemma \ref{prop1}. Clearly, the same will apply to $G'$ also. Hence by Lemma \ref{lmprop2}, we can infer that $G'$ is a co-bipartite CA graph.
\qed
\end{proof}
\begin{lemma}\label{lm2approx} 
The box representation $\mathcal{B}=\{I_1', I_2', \cdots, I_{b}', I'\}$, obtained in Line \ref{b'} of Algorithm \ref{algApprox} is a valid box representation of $G$ with $|\mathcal{B}| \le 2 \cdot box(G)+1$. 
\end{lemma}
\begin{proof}
  It is easy to see that $I'$ constructed in Line \ref{i'} of Algorithm \ref{algApprox} is a supergraph of $G$, since $I$ is an interval representation of the induced subgraph of $G$ on $B$ and $I'$ is an extension of $I$ on $V$. Since $\mathcal{B'}$ is a box representation of $G'$, each $I_i' \in \mathcal{B'}$, for $1 \le i \le b$ is a supergraph of $G'$ and in turn of $G$ too. $A$ is a clique in $G$ by definition. Consider any $(u, v') \notin E$ with $u \in A$ and $v' \in B$. Clearly, $(u, v') \notin E'$ as well and since $\mathcal{B'}$ is a box representation of $G'$, $\exists i$ such that $(u,v) \notin E(I_i')$ for some $1\le i\le b$. For any $(u', v') \notin E$ with $u', v' \in B$, we have $(u', v')\notin E(I')$. Thus, $G=I' \cap \displaystyle\bigcap_{1 \le i \le b}{I_{i}'}$. 

Thus, $\mathcal{B}=\{I_1', I_2', \cdots, I_{b'}', I'\}$ is a valid box representation for G of size $box(G') + 1$.  By Lemma \ref{lm4}, $box(G') \le 2 \cdot box(G)$, implying that $\mathcal{B}$ is of size at most  $2 \cdot box(G)+1$. 
\qed
\end{proof}
Lemma \ref{lm2approx} implies that $\mathcal{B}$ is a $(2+\frac{1}{k})$-factor approximate box representation where $k\ge 1$ is the boxicity of $G$. 
\section{Additive $2$-Factor Approximation for the Boxicity of Normal CA Graphs} \label{thAdd}
  We assume that a normal CA model $M(C, \mathcal A)$ of $G$ is given. An additive two factor approximation algorithm for computing a box representation of normal CA graphs is given in Algorithm \ref{algApprox2}. We can show that in $O(mn+kn^2)$ time, the algorithm outputs a near optimal box representation of $G$ where $n=|V(G)|$, $m=|E(G)|$ and $k=box(G)$. Refer to Appendix \ref{appendix1} for more details. If we just want to compute the approximate boxicity of $G$, it is enough to output $box(H)+2$, as proved below. This can be done in $O(mn+n^2)$.
\begin{algorithm}
     \LinesNumbered 
     \SetAlgoNoLine
     \DontPrintSemicolon
    \caption{Find a additive $2$-optimal box representation of given normal CA graph}
     \label{algApprox2}    
     \KwIn{A normal CA graph $G(V, E)$, with an NCA model $M(C, \mathcal A)$ of $G$}
     \KwOut{A box representation of $G$ of dimension at most $k+2$ where $k=box(G)$}    
     \lIf {$G$ is an interval graph}{
        Output an interval representation $I_G$ of $G$, Exit}\;
        Choose any point $p$ on the circle $C$; Let $A$ be the clique corresponding to $p$\label{p}\; 
        Let $p_1$ be the farthest clockwise end point of any arc passing through $p$\;
        Let $p_2$ be the farthest anticlockwise end point of any arc passing through $p$\;
        Let $q$ be a point on the arc $[p_1, p_2]$ with $q \ne p_1, p_2$\; \label{q}
        Let $B$ be the clique corresponding to $q$\;
        Let $H$ be the induced subgraph on $A \cup B$\\\label{h}\tcc*[f]{Clearly, $H$ is a co-bipartite CA graph}\;
        Find an optimum box representation $\mathcal{B'}=\{I_1, I_2, \cdots, I_h\}$ of $H$\\\tcc*[f]{Using the method described in Section \ref{sspoly}}\;
        \lFor{$i = 1$ to $h$}{
        Construct $I'_{i}$, the extension of $I_i$ on $V$\label{i}\;}
        Construct an interval representation $I_A$ for the induced subgraph on $V \setminus A$\\\tcc*[f]{Induced subgraph on $V \setminus A$ is an interval graph}\; 
        Construct $I'_A$, the extension of $I_A$ on $V$\label{a}\;
        Construct an interval representation $I_B$ for the induced subgraph on $V \setminus B$\\\tcc*[f]{Induced subgraph on $V \setminus B$ is an interval graph}\;
        Construct $I'_B$, the extension of $I_B$ on $V$\;
        Output  $\mathcal{B}=\{I_1', I_2', \cdots, I_{h}', I'_A, I'_B\}$ as the box representation of $G$\label{b}\; 
\end{algorithm} \\
\textbf{Proof of correctness: }
Since $G$ is a normal CA graph, the set of arcs passing through $p$ does not contain any circle cover pair of arcs. Therefore, $[p, p_1] \cup  [p_2, p]$ does not cover the entire circle $C$. So, any point in the arc $(p_1, p_2)$, in particular the point $q$ defined in Line \ref{q} of Algorithm \ref{algApprox2}, is not contained in any arc passing through $p$. It follows that $A \cap B=\emptyset$. Since $A$ and $B$ are cliques, $H$, the induced subgraph on $A \cup B$ is a co-bipartite CA subgraph of $G$. We can compute an optimum box representation $\mathcal{B'}$ of $H$ in polynomial time using the method described in Section \ref{sspoly}.  

$I_A$ and $I_B$ are interval graphs because they are obtained by removing vertices corresponding to arcs in $\mathcal A$ passing through points $p$ and $q$  respectively. Since $I_A$ is a supergraph of $G$ on $V\setminus A$ and $I'_A$ is the extension of $I_A$ on $V$, we can conclude that $I'_A$ is a super graph of $G$. Similarly, $I'_B$ is also a super graph of $G$. Since $\mathcal{B'}$ is a box representation of $H$, each $I_i \in \mathcal{B'}$ is a supergraph of induced subgraph $H$. Since $I_i'$ is the extension of $I_i$ on $V$, $I'_i$ is a super graph of $G$. 

 Consider $(u, v)\notin E$. \textbf{Case (i)} If $u, v \in V \setminus A$, by construction of $I'_{A}$, $(u, v)\notin E(I'_{A})$. \textbf{Case (ii)} If $u, v \in V \setminus B$, by construction of $I'_{B}$, $(u, v)\notin E(I'_{B})$. Remember that $A$ and $B$ are cliques. If both (i) and (ii) are false, then one of \{u, v\} is in $A$ and the other is in $B$. Since $\mathcal{B'}$ is a box representation of $H$, $(u,v) \notin E(I_i)$ for some $1\le i\le h = |\mathcal{B'}|$. By construction of $I'_i$, $(u,v) \notin E(I'_i)$ too. Hence, $G=I'_A \cap I'_B \cap \displaystyle\bigcap_{1 \le i \le h}{I_{i}'}$. Thus we get $\mathcal{B}=\{I'_{A}, I'_{B}, I'_{1}, I'_{2}, \cdots, I'_{h}\}$ is a valid box representation of $G$ of size $box(H) + 2$  which is at most $box(G)+2$, since $H$ is an induced subgraph of $G$.

In Algorithm \ref{algApprox2}, we assumed that an NCA model of the graph is given. This was required because recognizing NCA graphs in polynomial time is still an open problem. We can observe that though the algorithm of this section is given for normal CA graphs, it can be used for a wider class as stated below. 
\begin{theorem}
 If we are given a circular arc model $M(C, \mathcal{A})$ of $G$ with a point $p'$ on the circle $C$ such that the set of arcs passing through $p'$ does not contain a circle cover pair, then we can approximate the boxicity of $G$ within an additive $2$-factor in polynomial time using Algorithm \ref{algApprox2}.
\end{theorem}
\begin{proof}
  In Line \ref{p} of Algorithm \ref{algApprox2}, select $p'$ (guaranteed by the assumption of the theorem) as the point $p$. Such a point can be found in $O(n^2)$ time, if it exists. The rest of the algorithm is similar.
\qed
\end{proof}
  Though such a representation need not exist in general, it does exist for many important subclasses of of CA graphs and can be constructed in polynomial time; for example, for  proper CA graphs or normal helly CA graphs. In fact, for these classes, construction of a normal CA (NCA) model itself from their adjacency matrices can be done in polynomial time. 
\begin{corollary}
 Boxicity of any proper circular arc graph can be approximated within an additive $2$-factor in polynomial time. 
\end{corollary}
 \bibliographystyle{splncs03}
  
\newpage
\appendix
\section{Appendix 1}\label{appendix1}
\subsection*{Proof of Lemma \ref{CAclique}: }
  Let $G$ be a co-bipartite CA graph. Recall that a circular arc model of $G$ is constructable in linear time. In any circular arc model  $M(C,\mathcal{A})$ of a co-bipartite CA graph $G$, there are two points $p_1$ and $p_2$ on the circle $C$ such that every arc passes through at least one of them \cite{Tucker2,Lin09}. It is easy to see that these points can be identified in $O(n^2)$ time. Let the clique corresponding to $p_1$ be denoted as $A$. Let $B=V \setminus A$, which is clearly a clique, since the arcs corresponding to all vertices in $B$ pass through $p_2$. Let $|A|=n_1$ and $|B|=n_2$. Let vertices in $A$ be numbered $1, 2, \cdots, n_1$ and vertices of $B$ be numbered $1', 2', \cdots, n_2'$ according to the numbering scheme $NS(M, p_1)$ as described in the beginning of Section \ref{Num}. Clearly this numbering scheme satisfies Bi-Consecutive Adjacency Property by Lemma \ref{prop1}.
\qed
\subsection*{Proof of Lemma \ref{lmprop2}: }
The proof is by construction of a CA model $M(C,\mathcal{A})$ for $G$.\\
\textbf{Step 1:} Choose four distinct points $a, b, c, d$ in the clockwise order on $C$. Initially fix $s(i) = a$  for all $i \in A$ and $s(j') = c$  for all $j' \in B$. Choose $n_1$ distinct points $p_{n_1}$, $p_{n_1 - 1}$, $\cdots$, $p_1$ in the clockwise order on the arc $(a, b)$ and set $t(i)=p_i$ for all $ i \in A$. Choose $n_2$ distinct points $p_{n'_2}$, $p_{{n_2 - 1}'}$, $\cdots$, $p_{1'}$ in the clockwise order on the arc $(c, d)$ and set  $t(j')=p_{j'}$ for all $ j' \in B$. As of now, the family of arcs that we have constructed represents two disjoint cliques corresponding to $A$ and $B$.\\
\textbf{Step 2:} Now we will modify the start points of each arc as follows: Consider vertex $i \in A$. If $j' \in B$ is the highest numbered vertex in $B$ such that $i$ is adjacent to all $k'$ with $1' \le k' \le j'$, then set $s(i) = t(j')= p_{j'}$. Similarly, Consider vertex $j' \in B$. If $i \in A$ is the highest numbered vertex in $A$ such that $j'$ is adjacent to all $k$ with $1 \le k \le i$, then set  $s(j') = t(i)= p_i$. Notice that we are not making any adjacencies not present in $G$ between vertices of $A$ and $B$ in this step.

 Since $A$ and $B$ are cliques, what remains to prove is that if a vertex $i \in A$ is adjacent to a vertex $j' \in B$, their corresponding arcs overlap. Consider such an edge $(i, j')$. If $j'$ is adjacent to all $k$ such that $1\le k \le i$, we would have extended $s(j')$ to meet $t(i)$ in Step 2 above. If this does not occur, then by assumed Bi-Consecutive Adjacency Property, $i$ is adjacent to all $k'$ such that $1\le k' \le j'$. In this case, we would have extended $s(i)$ to meet $t(j')$ in Step 2. In both cases, the arcs corresponding to vertices $i$ and $j'$ overlap. We got a CA model of $G$ proving that $G$ is a CA graph.
\qed
\subsection*{Proof of Lemma \ref{lm4}: }
 Let $k$ be the boxicity of $G$ and $\{I_1, I_2, \cdots, I_k\}$ be an optimal box representation of $G$. For each $1 \le i \le k$, let   $l_i = \displaystyle\min_{u\in V}$ $l_u(I_i)$ and $r_i = \displaystyle\max_{u\in V}$ $r_u(I_i)$. Let $I_{i_1}$ be the interval graph obtained from $I_i$ by assigning the interval $\bigl[l_u(I_{i}),$ $r_u(I_{i})\bigr]$, $\forall u \in A$ and the interval $\bigl[l_i,$ $r_{v'}(I_{i})\bigr]$, $\forall v' \in B$. Let $I_{i_2}$ be the interval graph obtained from $I_i$ by assigning the interval $\bigl[l_u(I_{i}),$ $r_u(I_{i})\bigr]$, $\forall u \in A$ and the interval  $\bigl[l_{v'}(I_{i}),$ $r_i\bigr]$, $\forall v' \in B$.

  Note that, in constructing $I_{i_1}$ and  $I_{i_2}$ we have only extended some of the intervals of $I_i$ and therefore, $I_{i_1}$ and  $I_{i_2}$ are super graphs of $I$ and in turn of $G$. By construction, $B$ induces cliques in both  $I_{i_1}$ and  $I_{i_2}$, and thus they are supergraphs of $G_1$ too. 

 Now, consider $(u,v') \notin E$ with $u \in A$, $v' \in B$. Then either $r_{v'}(I_i) < l_u(I_i)$ or $r_u(I_i) < l_{v'}(I_i)$. If $r_{v'}(I_i) < l_u(I_i)$, then clearly the intervals $[l_i, r_{v'}(I_i)]$ and $[l_u(I_i), r_u(I_i)]$ do not intersect and thus $(u,v') \notin E(I_{i_1})$. Similarly, if $r_u(I_i) < l_{v'}(I_i)$, then $(u,v') \notin E(I_{i_2})$. If both $u, v \in A$ and $(u,v) \notin E$, then $\exists i$ such that  $(u,v) \notin E(I_i)$ for some $1\le i\le k$ and clearly by construction,  $(u,v) \notin E(I_{i_1})$ and  $(u,v) \notin E(I_{i_2})$.

  It follows that $G_1=\displaystyle\bigcap_{1 \le i \le k}{I_{i_1} \cap I_{i_2}}$ and therefore, $box(G_1) \le 2 \cdot box(G)$.                                                                                                                                                                                                                                                                                                                                                                                                                                                                                                                                                                                                                                                                                                                                                                                              \qed
\subsection*{Time Complexity of the Algorithm of Section \ref{sapprox}:}
Let $n=|V(G)|$ and $m=|E(G)|$. Whether the given graph is an interval graph can be determined in linear time. Given any CA graph $G(V, E)$, we can compute a CA model $M(C, \mathcal A)$ for $G$ in linear time \cite{Ross1}. A partition $(A, B)$ of the vertex set of $G$ as mentioned in the algorithm can be constructed in $O(n)$ time from the CA model $M$ of $G$. Construction of  $G'(V, E')$ from $G(V, E)$ can also be done in $O(n^2)$. Let $m_{_{AB}} = |\{ab' \in E(G) | a \in A$ and $b' \in B\}|$ and $m'_{_{AB}} = |\{ab' \in E(G') | a \in A$ and $b' \in B\}|$. In Section \ref{comp}, we discussed how to compute the boxicity of $G'$ in $O(en+n^2)$ time and an optimal box representation of co-bipartite CA graph $G'$ in $O(en+kn^2)$ where $k$ is the boxicity of $G'$ and $e=\min(m'_{_{AB}}, n_1 n_2 - m'_{_{AB}})$. Since by construction of $G'$, $m'_{_{AB}} = m_{_{AB}} \le m$, the time complexity is $O(mn+kn^2)$. The additional work for computing $I'$ in the construction of $\mathcal{B}$ can be done in $O(n)$. Thus, a near optimal box representation of $G$ is obtained in $O(mn+kn^2)$.
\subsection*{Time Complexity of the Algorithm of Section \ref{thAdd}:}
Whether the given graph is an interval graph can be determined in linear time. Choosing point $q$ can be done in $O(n)$ time. Construction of  $H$ from $G$ can also be done in $O(m+n)$.  Let $k=box(G)$ and $h=box(H)$. Since $H$ is an induced subgraph of $G$, $h \le k$. In Section \ref{comp}, we discussed how to compute the boxicity of $H$ in $O(en'+n'^2)$ and an optimal box representation of co-bipartite CA graph $H$ in $O(hen'+hn'^2)$, where $n'=|V(H)|\le n$ and $e=\min(|E(\overline H)|, |A|.|B| -|E(\overline H)|) \le m$. The additional work for computing the interval supergraphs $I'_A, I'_B$ is $O(n)$. Construction of $I'_{1}, I'_{2}, \cdots, I'_{h}$ from $I_{1}, I_{2}, \cdots, I_{h}$ requires only $O(h.n)=O(n^2)$ time in total. Thus, in $O(mn+kn^2)$ time, the algorithm outputs a near optimal box representation of $G$.
\section{Appendix 2 - Complexity of Computing the Boxicity and Optimal Box Representation of Co-bipartite CA graphs} \label{complexity}
 Let $G(V, E)$ be a co-bipartite CA graph with $|E| = m$ and $|V|= n$. Let $A \cup B$ be a partitioning of the vertex set $V$ as described in Lemma \ref{CAclique}, where $A$ and $B$ are cliques. Let $A=\{1, 2, \cdots, n_1\}$ and $B=\{1', 2', \cdots, n_2'\}$ be the associated numbering scheme. Let $m_{_{AB}} = |\{ab' \in E(G) | a \in A$ and $b' \in B\}|$ and $t = n_1 n_2 - m_{_{AB}}=|E(\overline G)|$. Let $e=\min(m_{_{AB}}, t)$ and $k=box(G)$. In this section, we will show an $O(en+n^2)$ algorithm to compute the boxicity and an $O(en+kn^2)$ algorithm to get an optimal box representation of $G$. Let $H=\overline G$. Recall that by Theorem \ref{th4}, $box(G)= \chi(H^*)$. Let $C_1, C_2, \cdots, C_k$ be the color classes in an optimal coloring of $H^*$. For $1\le i\le k$, let $C'_i$ be a maximal independent set containing $C_i$ and $E_i=\{e \in E(H)$: $e$ corresponds to a vertex in $C'_i\}$. By Theorem \ref{th4}, $\{G_i=\overline {H_i}$ : $H_i=(V, E_i)$, $1 \le i \le k \}$ gives an optimal box representation of $G$.
\subsection{Computing the Boxicity of $G$ in $O(en+n^2)$ Time} \label{appendix2}
 We call $ab'$ a \textit{non-edge} of $G$, if it is an edge of $H$. Recall that by Lemma \ref{lm1}, $H^*$ is a comparability graph. We had defined a transitive relation $\prec$ on $V(H^*)$ i.e on the non-edges of $G$, in the proof of Lemma \ref{lm1} as follows : $ab' \prec cd'$ if and only if $ a < c$ and $b' < d'$ and $\{a, b', c, d'\}$ induces a 4-cycle in $G$. Since $H^*$ is a comparability graph, any coloring satisfying the property that the color assigned to (the vertex corresponding to) a non-edge $cd'$ equals $\displaystyle\max_{e \in E(H) : e \prec cd'}{Color(e) + 1}$ is an optimum coloring \cite{Gallai} of $H^*$. We refer to this as greedy strategy in our further discussion. For convenience, hereafter we refer to the coloring of a vertex of $H^*$ as coloring of the corresponding non-edge of $G$. 

Assume that the colors available are $1, 2, \cdots$. The following definitions are with respect to $G$. For $X \subseteq V$, let $N_{_{X}}(v)$ represent the set of neighbors of $v$ in $X$ and $\widehat N_{_{X}}(v)=X \setminus N_{_{X}}(v)$.  Similarly, for $S \subseteq V$, $N_{_{X}}(S)= \displaystyle \bigcup_{v\in S} N_{_{X}}(v)$ and $\widehat N_{_{X}}(S)=\displaystyle \bigcup_{v\in S} \widehat N_{_{X}}(v)$. Let $deg_{_{X}}(v)$ denote $|N_{_{X}}(v)|$. The linked lists corresponding to $N_{_{B}}(v)$ and $\widehat N_{_{B}}(v)$ for each $v \in A$ and $N_{_{A}}(v')$ and $\widehat N_{_{A}}(v')$ for each $v' \in B$, with their entries sorted with respect to the numbering scheme described in the above paragraph, can be constructed from the adjacency list of $G$. This can be done in overall $O(n^2)$ time. We will assume that lists $N_{_{A}}, \widehat N_{_{A}}, N_{_{B}}, \widehat N_{_{B}}$ are global data structures. 

For $x\in A$, we color the non-edges incident on $x$ by invoking Algorithm \ref{alg1}, for $x=1, 2, \cdots, n_1$ in that order. For the convenience of our analysis, we refer to an invocation of Algorithm \ref{alg1} for vertex $x$ as \textit{the processing of $x$}. Note that, by the time a non-edge $xy'$ of $G$ is considered for coloring, i.e, during the processing of $x$, all non-edges $tu'$ such that $tu'\prec xy'$ are already colored, since, by the definition of $\prec$, $t < x$ and $t$ is processed before $x$. Consider a non-edge $xy'$ of $G$. Let $F_{xy'}=F_{y'}= \{ab'\in E(H):ab' \prec xy' \}$. According to the greedy algorithm, the non-edge $xy'$ of $G$ has to get the color $maxcolor(F_{y'}) + 1$, where  $maxcolor(F_{y'})= \displaystyle\max_{ab' \in F_{y'}}{Color(ab')}$. 

The next question is how to find $maxcolor(F_{y'})$ efficiently. For that we need to understand the set $F_{y'}$ more closely. Let $P=\{a \in N_{_{A}}(\widehat N_{_{B}}(x))|a < x\}$ and  $Q=\{ b' \in N_{_{B}}(x)|b'< \min{\widehat N_{_{B}}(x)}\}$. 
\begin{claim}
 ${F_{y'}=\displaystyle\biguplus_{a \in N_{_{A}}(y')\cap P} \{ab'\in E(H): b' \in Q\}}=\{ab'\in E(H):a \in N_{_{A}}(y')\cap P$ and $b' \in Q\}$.
\end{claim}
\begin{proof}
 Since $F_{y'}=\{ab'\in E(H):ab' \prec xy' \}$, we need to show that for any $ab' \in E(H)$, $ab'\prec xy'$ if and only if $a \in N_{_{A}}(y')\cap P$ and $b'\in Q$. Recall that $ab' \prec xy'$ if and only if $a<c$, $b'< d'$ and $\{a, b', c, d'\}$ induces a 4-cycle in $G$. Observe that $N_{_{A}}(y')\cap P=\{a \in N_{_{A}}(y') : a < x\}$. It is easy to see that if $a \in N_{_{A}}(y')\cap P$ and $b'\in Q$, then $ab'\prec xy'$.  

To prove the other direction, assume that $ab' \prec xy'$. Then we have $a \in N_{_{A}}(y')$, $a < x$ and therefore, $a \in N_{_{A}}(y')\cap P$. Similarly, $b' \in N_{_{B}}(x)$, $b'< y'$. Suppose $\min{\widehat N_{_{B}}(x)} < b'$. Since the numbering scheme satisfies Bi-Consecutive Adjacency Property, $xb' \in E(G)$ implies that either $(x, \min{\widehat N_{_{B}}(x)})\in E(G)$ or $ab'\in E(G)$, which is a contradiction. Therefore $b'< \min{\widehat N_{_{B}}(x)}$ and therefore, $b' \in Q$. 
\qed
\end{proof}
 By the above claim, $maxcolor(F_{y'})=\displaystyle\max_{a \in N_{_{A}}(y') \cap P}\{\displaystyle\max_{b' \in Q, ab' \in E(H)}{\{Color(ab')\}}\}$. But, if we have to do this computation separately for each $y' \in \widehat N_{_{B}}(x)$, then for any $a \in P$ which is in $N_{_{A}}(y')$ of more than one $y'$, the computation of $maxcolor(a)=\displaystyle\max_{b' \in Q, ab' \in E(H)}{\{Color(ab')\}}$ has to be repeated. To avoid this repetition, in Algorithm \ref{alg1} we process all non-edges incident at a vertex $x \in A$ in parallel as follows. In Lines \ref{sf} to \ref{ef} of Algorithm \ref{alg1}, for each $a \in P$, $maxcolor(a)+1$ is computed and stored. This is referred to as Type 1 work in the algorithm. For each $y' \in \widehat N_{_{B}}(x)$, Lines \ref{sf1} to \ref{ef1} referred to as Type 2 work in Algorithm \ref{alg1} computes $maxcolor(F_{y'}) + 1$ using the values of $maxcolor(a)+1$ already computed and stored as part of Type 1 work. In the process, for each $y' \in \widehat N_{_{B}}(x)$, the algorithm assigns the color $maxcolor(F_{y'}) + 1$ to $xy'$, which is the optimum color suggested by the greedy strategy.

\begin{algorithm}
     \LinesNumbered 
     \SetAlgoNoLine
     \DontPrintSemicolon
 \caption{Computing colors of non-edges incident on vertex $x \in A$}
\label{alg1} 
     \KwIn{$x \in A$ }
    \KwOut{$Color(xy')$ for each $y' \in \widehat N_{_{B}}(x)$}
        \tcc{Type 0 work : Lines \ref{Ins} to \ref{Ine} - Initializations}
         \tcc{Let $P=\{a \in N_{_{A}}(\widehat N_{_{B}}(x)) | a < x\}$}
          {For $1\le a \le n_1$, let $A_P[a]=0$ initially. For each $a \in P$, set $A_P[a]=1$ and $color[a] = 0$\label{Ins}}\; 
    {Compute  $Q=\{b' \in N_{_{B}}(x)| b'< p'\}$, where $p' =\min{(\widehat N_{_{B}}(x))}=\widehat N_{_{B}}(x)[1]$ and initialize $ptr1[b']=$ start of  $\widehat N_{_{A}}(b')$ for $b' \in Q$\label{s2}}\;
         {Compute $R=\widehat N_{_{B}}(x)$ and initialize $ptr2[r']=$ start of  $N_{_{A}}(r')$ for $r' \in R$ \label{Ine}}\;
       \For{$cur = 1$ to $n_1$}{                           
            \If {$A_{P}[cur]=1$}{
                \tcc{Type 1 work : Lines \ref{sf} to \ref{ef} - Computing $color[cur]= 1+ $ the maximum color given to a non-edge between $cur$ and $Q$}        
                 \For {each $q'$ in $Q$\label{sf}}{ 
                        \While {$\widehat N_{_{A}}(q')[ptr1[q']]< cur $ and not list-end of $\widehat N_{_{A}}(q')$\label{wh1}} 
                         {Increment the pointer $ptr1[q']$\;}                       
		      \uIf{$\widehat N_{_{A}}(q')[ptr1[q']] = cur$}{
			  $color[cur] = Max (color[cur],$ $Color(cur$ $q') + 1)$ \label{l11}  \tcc*[assign]{non-edge $(cur$ $q')$ is already colored} }
                      \lElseIf{list-end of $\widehat N_{_{A}}(q')$}
                       { delete $q'$ from $Q$\;}
                   }\label{ef}  
		\tcc{Type 2 work : Lines \ref{sf1} to \ref{ef1} - Identify non-edges at $x$ affected by non-edges between $cur$ and $Q$ and update their colors if necessary}  
                   \For {each $r'$ in $R$}{\label{sf1}
	                 \While {$N_{_{A}}(r')[ptr2[r']]< cur $ and not list-end of $N_{_{A}}(r')$\label{wh2}} 
                          {Increment the pointer $ptr2[r']$\;}
		          \uIf{$N_{_{A}}(r')[ptr2[r']] = cur$}{ 
                          \lIf{$Color(xr')< color[cur]$  \label{line20}}{
			  $Color(xr') = color[cur]$\;}
		          } 
                        \lElseIf{list-end of $N_{_{A}}(r')$}
                          { delete $r'$ from $R$\;}
                      }\label{ef1}
                 }
	     }
           \end{algorithm}
\begin{lemma}
 Time spent over all invocations of Algorithm \ref{alg1} is $O(en+n^2)$.
\end{lemma}
\begin{proof}
Let $m_{_{AB}} = |\{ab' \in E(G) | a \in A$ and $b' \in B\}|$ and $t = n_1 n_2 - m_{_{AB}}=|E(\overline G)|$. Let $e=\min(m_{_{AB}}, t)$ and $k=box(G)$.

 Type 0 work (Lines \ref{Ins} to \ref{Ine}) computes lists $Q$ and $R$ as defined in the algorithm and also an indicator array $A_P$ of $P$.  $Q$ and $R$ can be represented as doubly linked lists. Initializations in Line \ref{Ins} can be achieved in $O(en+n^2)$ time as follows: $A_P$ can be initialized to $0$ in $O(n)$ once for each $x \in A$. The total time for this work is $O(n^2)$. Each $y' \in B$ spends at most $|N_{_{A}}(y')|$ time for the processing of each $x \in \widehat N_{_{A}}(y')$. Thus the total time spent by all $y' \in B$ together for this initialization is $\displaystyle\sum_{y'\in B} {deg_{_{A}}(y')(n_1 -deg_{_{A}}(y'))}$ $= O(en)$.  Initialization of $Q$ in Line \ref{s2} can be done in $O(deg_{_{B}}(x))$. Summing over all $x \in A$, this amounts to $O(m_{_{AB}})=O(m)$ work. Similarly, for initializing $R$ in Line \ref{Ine}, we need $O(n_2 - deg_{_{B}}(x))$. Summed over all $x \in A$, this amounts to $O(t)$ work. Adding all the above, total cost of Type 0 work (over all invocations of Algorithm \ref{alg1}) is $O(en+n^2+m+t)=O(en+n^2)$, since $m+t=O(n^2)$.

 Let us calculate the total cost spent in Type 1 work. Note that each element $q' \in Q$ remembers the pointer position $ptr1[q']$. This means that  $ptr1[q']$ continues from where it stopped in the current iteration, while doing the Type 1 work of the next element of $P$. Therefore, pointer $ptr1[q']$ moves at most $n_1-deg_{_{A}}(q')$ times for each $q' \in Q$. When $ptr1[q']$ reaches the end of list $\widehat N_{_{A}}(q')$, $q'$ is deleted from the linked list $Q$. This makes sure that Line \ref{wh1} is repeated just $O(n_1-deg_{_{A}}(q'))$ times for each $q' \in Q$. Each $q'$ executes Line \ref{l11} whenever $cur \in \widehat N_{_{A}}(q')$. This also happens $n_1 - deg_{_{A}}(q')$ times during the processing of each $x$ where $x \in N_{_{A}}(q')$. Hence the total cost for Type 1 work (over all invocations of Algorithm \ref{alg1}) is $\displaystyle\sum_{b'\in B}{(n_1 -deg_{_{A}}(b')). deg_{_{A}}(b')} = O(en)$.

Now consider Type 2 work. Note that each element $r'\in R$ remembers the pointer position $ptr2[r']$. This means that  $ptr2[r']$ continues from where it stopped in the current iteration while doing the Type 2 work of the next element of $P$. Therefore, pointer $ptr2[r']$ moves at most $deg_{_{A}}(r')$ times for each $r' \in  R=\widehat N_{_{B}}(x)$. When $ptr2[r']$ reaches end of list $N_{_{A}}(r')$, $r'$ is dropped from the linked list $R$. This makes sure that Line \ref{wh2} is repeated only $O(deg_{_{A}}(r'))$ times for each $r' \in R$, while processing an $x$ such that $x \in \widehat N_{_{A}}(r')$. Also, Line \ref{line20} is executed only when $cur \in N_{_{A}}(r')$. This happens $deg_{_{A}}(r')$ times during the processing of each $x$, where $x \in \widehat N_{_{A}}(r')$. Summing up, the total cost for Type 2 work (over all invocations of Algorithm \ref{alg1}) is $\displaystyle\sum_{b'\in B}{deg_{_{A}}(b'). (n_1 -deg_{_{A}}(b'))} = O(en)$.

Thus the total cost spent over all invocations of Algorithm \ref{alg1} is $O(en+n^2)$ as claimed.
\qed 
\end{proof}
\subsection{Expanding Color Classes of $H^*$ to Maximal Independent Sets in $O(en+kn^2)$}
In this section, we describe an algorithm which computes an optimal box representation $\mathcal{B}=\{G_1, G_2, \cdots, G_k\}$ of $G$ as explained in the beginning of Section \ref{complexity}, where $k$ is the maximum number of colors used by the algorithm of Section \ref{appendix2}. Since edges of $H$ correspond to vertices of $H^*$ by defintion, we do not differenciate between a non-edge $ab' \in E(H)$ and its corresponding vertex in $H^*$ in the following discussions in this section. Recall from the beginning of Section \ref{complexity} that we can compute $\mathcal{B}$ by computing  $C'_i$, for $1\le i \le k$, where $C'_i$ is a maximal independent set containing $C_i$ - the $i^{th}$ color class in the optimal coloring of $H^*$ obtained by the algorithm of Section \ref{appendix2}. The following lemma suggests one way to compute these maximal independent sets. 
\begin{lemma}
 Let $C_i$ be the $i^{th}$ color class in the optimal coloring of $H^*$ obtained by the algorithm of Section \ref{appendix2}. Let $S_i= \displaystyle \bigcup_{1\le j\le i} C_i$ and let $MaxS_i$ be the set of maximal elements of $(S_i, \prec)$, i.e, $MaxS_i = \{ ab' \in S_i : \nexists cd' \in S_i$ with $ab' \prec cd'\}$. Then $MaxS_i$ is a maximal independent set in $H^*$ containing $C_i$.
\end{lemma}
\begin{proof}
 $MaxS_i$, being the set of maximal elements of $(S_i, \prec)$, forms an independent set in $H^*$. 
 Recall that, as per our algorithm, for any  $ab' \in E(H)$, $Color(ab') = \displaystyle\max_{e \in E(H) : e \prec ab'}{Color(e) + 1}$. Consider $ab' \in C_i$. If $\exists cd' :  ab' \prec cd'$, then $Color(cd') > Color(ab') = i$ and therefore, $cd' \notin S_i$. Hence, by the definition of $MaxS_i$, $ab'  \in  MaxS_i$. Thus, $C_i \subseteq MaxS_i$. 

Consider any $ab' \notin MaxS_i$. Either $ab' \in (S_i \setminus MaxS_i)$ or $ab' \notin S_i$. In the former case, $\exists cd' \in MaxS_i$ with $ab' \prec cd'$. In the latter case, when $ab' \notin S_i$, $Color(ab') > i$ and it is easy to see from our coloring strategy that $\exists cd' \in C_i \subseteq MaxS_i$ with $cd'  \prec ab'$. Therefore, in both cases, if $ab'$ is added to  $MaxS_i$, it will no longer be an independent set. Thus, $MaxS_i$ is a maximal independent set containing $C_i$.  
\qed
\end{proof} 
The next question is to efficiently compute $MaxS_i$, for $1\le i\le k$. For this purpose, we introduce the following definition.
\begin{definition}\label{lnext}
For each $ab' \in E(H)$, let
\begin{displaymath}
Next(ab') = \left\{ \begin{array}{ll}
 \displaystyle \min_{e \in E(H), ab' \prec e } \{Color(e)\}, {\mathrm{\hspace{1cm} if \ }\exists e \in E(H) \mathrm{\ such \ that \ }ab' \prec e}\\
 k+1,\mathrm{ \hspace{3.7cm} otherwise}
  \end{array} \right.
\end{displaymath}
\end{definition}
\begin{lemma}
 For $1 \le i \le k$, $MaxS_i = \{ab' \in S_i$ : $Next(ab') > i\}$
\end{lemma}
\begin{proof}
  If $ab' \notin MaxS_i$, then $\exists cd' \in S_i$ with $ab' \prec cd'$. It will follow that $Next(ab') \le Color(cd') \le i$. Conversely, if $Next(ab') \le i$, then $\exists cd' \in S_i$ with $ab' \prec cd'$ and hence $ab' \notin MaxS_i$. 
\qed
\end{proof}
Let $Next(ab')$ for all $ab' \in E(H)$ be initialized to $k+1$. This can be done in $O(|E(H)|)$ = $O(n^2)$. Consider the following strategy. Take a non-edge $e \in E(H)$ and update $Next(ab')$ of all $ab' \prec e$ with  $\min( Next(ab'), Color(e))$. When we have repeated this for all $e \in E(H)$, it is easy to see that the values of $Next(ab')$ for every $ab' \in E(H)$ will satisfy Definition \ref{lnext}. 

Here, we show that we can do this in $O(en+n^2)$. In order to achieve this, we process the non-edges incident at a vertex $x \in A$ together, in an invocation of Algorithm \ref{alg2} - hereafter called the processing of $x$. During the processing of $x$, each non-edge $xy'$ incident at $x$ updates $Next(ab')$ of all $ab' \prec xy'$ with $\min(Next(ab'), Color(xy'))$. We will process $x=1, 2, \cdots, n_1$ in that order. The data structures used are similar to those used for Algorithm \ref{alg1}. 

Consider an $x \in A$. As in Section \ref{appendix2}, let $P=\{a \in N_{_{A}}(\widehat N_{_{B}}(x))|a < x\}$, $Q=\{ b' \in N_{_{B}}(x)|b'< \min{\widehat N_{_{B}}(x)}\}$ and $F_{y'}= \{ab'\in E(H):ab' \prec xy' \}$. Let $T_x = \displaystyle \bigcup_{y' \in \widehat N_{_{B}}(x)} \{ab' \in E(H) : ab' \prec xy'\} = \displaystyle \bigcup_{y' \in \widehat N_{_{B}}(x)} F_{y'}$. Observe that, for any $ab' \in E(H) \setminus T_x$, the value of $Next(ab')$ does not depend on the colors assigned to non-edges incident at $x$. By the claim proved in Section \ref{appendix2}, $F_{y'}=\{ab'\in E(H):a \in N_{_{A}}(y')\cap P$ and $b' \in Q\}$. Hence, $T_x = \{ab'\in E(H):a \in P$ and $b' \in Q\}$. Hence, during the processing of $x \in A$, we just need to update the $Next$ values of non-edges between $P$ and $Q$ only. 

Consider any $ab' \in T_x$. The set of non-edges incident at $x$ whose colors can affect the value of $Next(ab')$ belong to the set $\{xy' : y' \in \widehat N_{_{B}}(x)$ and $ab' \prec xy' \}$ = $\{xy' : y' \in \widehat N_{_{B}}(x) \cap N_{_{B}}(a)\}$. Notice that for any fixed $a \in P$, this set is independent of which $b' \in Q$ is being considered. Let us denote this set by $U_a$. For any non-edge $ab' \in E(H)$ with $a \in P$ and $b' \in Q$, $Next(ab') \le \displaystyle \min_{e \in U_a}\{Color(e)\}$. Hence, we can make the following inference, which is critical for the efficiency of Algorithm \ref{alg2}:\\
\textbf{{Fact. }}
For a fixed vertex $a \in P$, for any non-edge $ab' \in E(H)$ between $a$ and $Q$, we just need to update $Next(ab')$ with $\min(Next(ab'), MinColor[a])$, where $MinColor[a]= \displaystyle \min_{e \in U_a}\{Color(e)\}$, irrespective of which $b' \in Q$ is involved. 

In Algorithm \ref{alg2}, Type 1 work computes $MinColor[a]$ for every $a \in P$. Type 2 work updates $Next(ab')$, for each $ab' \in T_x$ with $\min(Next(ab'),MinColor[a])$. By the time we have processed all $x \in A$, all non-edges $e \in E(H)$ get processed and hence $Next(ab')$ for each $ab' \in E(H)$ is correctly computed.
\begin{algorithm}
     \LinesNumbered 
     \SetAlgoNoLine
     \DontPrintSemicolon
     \caption{Each non-edge $xy'$ incident at vertex $x \in A$ updates $Next(ab')$ of all non edges $ab' \prec xy'$}
     \label{alg2} 
       \KwIn{$x \in A$}
       \KwOut{The updated $Next(ab')$ for each non edges $ab' \prec xy'$ where $y' \in \widehat N_{_{B}}(x)$}    
        \tcc{Type 0 work : Lines \ref{Ins1} to \ref{Ine1} - Initializations}
         \tcc{Let $P=\{a \in N_{_{A}}(\widehat N_{_{B}}(x)) | a < x\}$}
          {For $1\le a \le n_1$, let $A_P[a]=0$ initially. For each $a \in P$, set $A_P[a]=1$ and $MinColor[a]=k+1$\label{Ins1} }\;
    {Compute  $Q=\{b' \in N_{_{B}}(x)| b'< p'\}$, where $p' =\min{(\widehat N_{_{B}}(x))}=\widehat N_{_{B}}(x)[1]$ and initialize $ptr2[b']=$ start of  $\widehat N_{_{A}}(b')$ for $b' \in Q$\label{s22}}\;
         {Compute $R=\widehat N_{_{B}}(x)$ and initialize $ptr1[r']=$ start of  $N_{_{A}}(r')$ for $r' \in R$ \label{Ine1}}\;
         \For{$cur= 1$ to $n_1$}{                           
            \If {$A_{P}[cur]=1$}{
                    \tcc{Type 1 work : Lines \ref{sf2} to \ref{ef2} - Computing $MinColor[cur]=$ the minimum color given to a non-edge between $x$ and $N_{_{B}}(cur) \cap R$}        
                    \For {each $r'$ in $R$\label{sf2}}{
	                 \While {$N_{_{A}}(r')[ptr1[r']]< cur $ and not list-end of $N_{_{A}}(r')$\label{wh11}} 
                          {Increment the pointer $ptr1[r']$\;}
		          \uIf{$N_{_{A}}(r')[ptr1[r']] = cur$}{ 
			  $MinColor[cur] = \min(MinColor[cur],Color(xr'))$\label{line11}\;}
		          \lElseIf{list-end of $N_{_{A}}(r')$}
                          { delete $r'$ from $R$\;}			  
                       }\label{ef2}
              \tcc{Type 2 work : Lines \ref{sf3} to \ref{ef3} - Update $Next$ of non-edges between $cur$ and $Q$}  
                 \For {each $q'$ in $Q$\label{sf3}}{ 
                        \While {$\widehat N_{_{A}}(q')[ptr2[q']]< cur $ and not list-end of $\widehat N_{_{A}}(q')$\label{wh21}} 
                         {Increment the pointer $ptr2[q']$\;}                       
		      \uIf{$\widehat N_{_{A}}(q')[ptr2[q']] = cur$}{
                          $Next(cur$ $q') = \min(Next(cur$ $q'), MinColor[cur])$ \label{l20}
		       }
                     \lElseIf{list-end of $\widehat N_{_{A}}(q')$}
                       { delete $q'$ from $Q$\;}
                   }\label{ef3}

                  }  
               }     
     \end{algorithm}
\begin{lemma}
  Time spent over all invocations of Algorithm \ref{alg2} is $O(en+n^2)$.
\end{lemma}
\begin{proof}
 Type 0 work done by Algorithm \ref{alg2} (Lines \ref{Ins1} to \ref{Ine1}) is similar to the Type 0 work of Algorithm \ref{alg1} and hence the total cost of Type 0 work over all invocations of Algorithm \ref{alg2} is $O(en+n^2)$.

Let us calculate the total cost spent in Type 1 work. Note that each element $r'\in R$ remembers the pointer position $ptr1[r']$. This means that  $ptr1[r']$ continues from where it stopped in the current iteration while doing the Type 1 work of the next element of $P$. Therefore, pointer $ptr1[r']$ moves at most $deg_{_{A}}(r')$ times for each $r' \in  R=\widehat N_{_{B}}(x)$. When $ptr1[r']$ reaches end of list $N_{_{A}}(r')$, $r'$ is dropped from the linked list $R$. This makes sure that Line \ref{wh11} is repeated only $O(deg_{_{A}}(r'))$ times for each $r' \in R$, while processing an $x$ such that $x \in \widehat N_{_{A}}(r')$. Also, Line \ref{line11} is executed only when $cur \in N_{_{A}}(r')$. This happens $deg_{_{A}}(r')$ times during the processing of each $x$, where $x \in \widehat N_{_{A}}(r')$. Summing up, the total cost for Type 1 work (over all invocations of Algorithm \ref{alg2}) is $\displaystyle\sum_{b'\in B}{deg_{_{A}}(b'). (n_1 -deg_{_{A}}(b'))} = O(en)$.

Now consider Type 2 work. Note that each element $q' \in Q$ remembers the pointer position $ptr2[q']$. This means that  $ptr2[q']$ continues from where it stopped in the current iteration, while doing the Type 2 work of the next element of $P$. Therefore, pointer $ptr2[q']$ moves at most $n_1-deg_{_{A}}(q')$ times for each $q' \in Q$. When $ptr2[q']$ reaches the end of list $\widehat N_{_{A}}(q')$, $q'$ is deleted from the linked list $Q$. This makes sure that Line \ref{wh21} is repeated just $O(n_1-deg_{_{A}}(q'))$ times for each $q' \in Q$. Each $q'$ executes Line \ref{l20} whenever $cur \in \widehat N_{_{A}}(q')$. This also happens $n_1 - deg_{_{A}}(q')$ times during the processing of each $x$ where $x \in N_{_{A}}(q')$. Hence the total cost for Type 2 work (over all invocations of Algorithm \ref{alg2}) is $\displaystyle\sum_{b'\in B}{(n_1 -deg_{_{A}}(b')). deg_{_{A}}(b')} = O(en)$.

Thus the total time spent over all invocations of Algorithm \ref{alg2} is $O(en+n^2)$ as claimed.
\qed
\end{proof}

Once $Next(ab')$ for (vertex of $H^*$ corresponding to) each $ab' \in E(H)$ is correctly computed by invoking Algorithm \ref{alg2} for each $x \in A$, we compute $MaxS_i = \{ab' \in S_i$ : $Next(ab') > i\}= \{ab' \in E(H)$ : $Color(ab') \le i$ and $Next(ab') > i\}$, for $1 \le i \le k$. This can be done in overall $O(k.|E(H)|)$=$O(k. n^2)$ time. For $1\le i\le k$, let $E_i=\{e \in E(H)$: $e$ corresponds to a vertex of $H^*$ in $MaxS_i\}$. As mentioned in the beginning of Section \ref{complexity}, $\{G_i=\overline {H_i}$ : $H_i=(V, E_i)$, $1 \le i \le k \}$, gives an optimal box representation of $G$. Since each $G_i$ can be computed from $E_i$ in $O(n^2)$, the overall running time of computing the box representation using $E_i$s is $O(kn^2)$. 

Thus, the total time used for the algorithm for computing an optimal box representation of $G$ is $O(en+kn^2)$ as claimed.
\end{document}